\documentclass{article}

\usepackage[T1]{fontenc}
\usepackage[english]{babel}
\usepackage[letterpaper,top=2cm,bottom=2cm,left=3cm,right=3cm,marginparwidth=2cm]{geometry}
\usepackage{amsmath}
\usepackage{amssymb}
\usepackage{amsthm}
\usepackage{bbm}
\usepackage{mathtools}
\usepackage{enumitem}
\usepackage{graphicx}
\usepackage{xparse}
\usepackage{siunitx}
\usepackage{pdflscape}
\usepackage{todonotes}
\usepackage[colorlinks=true, allcolors=blue]{hyperref}
\usepackage[capitalize,noabbrev]{cleveref}
\usepackage{fontawesome} % for \faClockO

\DeclareDocumentCommand\zerovec{o}{\IfNoValueTF{#1}{\mathbb{O}}{\mathbb{O}_{#1}}}
\DeclareDocumentCommand\unitvec{m}{\mathbbm{e}^{#1}}
\DeclareDocumentCommand\R{}{\mathbb{R}}
\DeclareDocumentCommand\transpose{m}{#1^{\intercal}}
\DeclareDocumentCommand\rank{o}{\operatorname{rank}\IfValueTF{#1}{\left(#1\right)}{}}

\DeclareDocumentCommand\orderO{o}{\mathcal{O}\IfValueTF{#1}{\left(#1\right)}{}}
\DeclareDocumentCommand\orderTheta{o}{\Theta\IfValueTF{#1}{\left(#1\right)}}
\DeclareDocumentCommand\orderOmega{o}{\Omega\IfValueTF{#1}{\left(#1\right)}{}}

\DeclareDocumentCommand\reviewComment{mm}{\todo[inline,color=red!50!white]{\textbf{Reviewer:} #1 \textbf{Response:} #2}}

\DeclareDocumentCommand\reviewComment{mm}{}

\newtheorem{theorem}{Theorem}
\newtheorem{proposition}[theorem]{Proposition}

\newtheorem{corollary}[theorem]{Corollary}

\title{The Binary Linearization Complexity of Pseudo-Boolean Functions}
\author{Matthias Walter\footnote{Department of Applied Mathematics, University of Twente, The Netherlands. E-mail: \texttt{m.walter@utwente.nl}}}

\begin{document}

\maketitle

\begin{abstract}
  We consider the problem of linearizing a pseudo-Boolean function $f : \{0,1\}^n \to \R$ by means of $k$ Boolean functions.
  Such a linearization yields an integer linear programming formulation with only $k$ auxiliary variables.
  This motivates the definition of the linearization complexity of $f$ as the minimum such $k$.
  Our theoretical contributions are the proof that random polynomials almost surely have a high linearization complexity and characterizations of its value in case we do or do not restrict the set of admissible Boolean functions.
  The practical relevance is shown by devising and evaluating integer linear programming models of two such linearizations for the low auto-correlation binary sequences problem.
  Still, many problems around this new concept remain open.
\end{abstract}

\section{Introduction}
\label{sec_intro}

\DeclareDocumentCommand\lc{m}{\mathop{lc}_{#1}}
\DeclareDocumentCommand\G{}{\mathcal{G}}
\DeclareDocumentCommand\B{}{\mathcal{B}}
\DeclareDocumentCommand\C{}{\mathcal{C}}
\DeclareDocumentCommand\M{}{\mathcal{M}}

We consider \emph{pseudo-Boolean functions}, i.e., $f : \{0,1\}^n \to \R$.
Such functions arise at the intersection of \emph{Constraint Programming} and \emph{Operations Research}.
The corresponding \emph{pseudo-Boolean optimization problem} of minimizing $f$ over $\{0,1\}^n$ is known to be NP-hard since it subsumes the maximum cut problem~\cite{BorosH02,Karp72}.
Besides being subject to minimization, pseudo-Boolean functions appear frequently in problem constraints, e.g., in satisfiability problems.
It is known well that every pseudo-Boolean function has a unique representation by means of a multilinear polynomial.
Above-mentioned optimization problems can be reduced to integer \emph{linear} programming problems by introducing auxiliary variables, e.g., for every monomial of its multilinear representation.
Since there may be a huge number of monomials, the natural question for alternative ways of linearizing using fewer auxiliary variables arises.
The main research question addressed by this paper is that for the minimum number of such auxiliary variables under the additional restriction that these are binary variables.

To this end, a \emph{linearization} of $f$ is defined by functions $g_1,g_2,\dotsc,g_k : \{0,1\}^n \to \R$ and parameters $a \in \R^n$, $\beta \in \R$ and $b_i \in \R$ for $i=1,2,\dotsc,k$ such that
\begin{equation}
  f(x) = \transpose{a}x + \beta + \sum_{i=1}^k b_i g_i(x) \label{eq_linearization}
\end{equation}
holds for all $x \in \{0,1\}^n$.
Its \emph{size} is the number $k$ of such functions.
The \emph{linearization complexity of $f$ with respect to a family $\G$} of functions is defined as the minimum size of a linearization with $g_i \in \G$ for all $i$, and is denoted by $\lc{\G}(f)$.

In this paper we focus on \emph{binary} linearizations, which are those where each function $g \in \G$ is Boolean, i.e., $g : \{0,1\}^n \to \{0,1\}$ holds.
For a binary linearization of size $k$ there exists an integer linear programming (IP) formulation for modeling the relationship of $x$ and $f(x)$ with $n+k$ variables which works as follows.
In addition to the variables $x \in \{0,1\}^n$ we introduce the binary variables $y \in \{0,1\}^k$ that shall represent the values $g_i(x)$.
By substituting $g_i(x)$ with $y_{i}$, the function value $f(x)$ is expressed as an affine function in the variables $(x,y)$.
The following constraints ensure that $y_i = g_i(x)$ holds:
\begin{subequations}
  \label{eq_no_good}
  \begin{alignat}{7}
    \sum_{j : \bar{x}_j = 0} x_j &\;+\;& \sum_{j : \bar{x}_j = 1} (1-x_j) &\;+\;& y_i &\geq 1 &\quad& \forall \bar{x} \in \{0,1\}^n : g_i(\bar{x}) = 1, \\
    \sum_{j : \bar{x}_j = 0} x_j &\;+\;& \sum_{j : \bar{x}_j = 1} (1-x_j) &\;+\;& (1-y_i) &\geq 1 &\quad& \forall \bar{x} \in \{0,1\}^n : g_i(\bar{x}) = 0.
  \end{alignat}  
\end{subequations}
Note that~\eqref{eq_no_good} consists of $2^n$ inequalities.
The associated separation problem can be solved with linearly many evaluations of $g_i$ plus linear computation time by computing, for given $\hat{x} \in [0,1]^n$ the closest vertex $\bar{x}^0 \in \{0,1\}^n$ of the cube and all vectors $\bar{x}^j \in \{0,1\}^n$ that can be obtained from $\bar{x}^0$ by flipping coordinate $j$, computing $g_i(\bar{x}^j)$ and testing~\eqref{eq_no_good} for $j=0,1,2,\dotsc,n$.
However, it is worth noting that the LP bounds obtained from such a formulation are typically very weak.
A special case in which~\eqref{eq_no_good} yields a perfect formulation is that of a parity indicator variable, that is, if $g_1(x) = 1$ holds if and only if $\sum_{i=1}^n x_i$ is odd~\cite{Jeroslow75}.
Here, perfect means that the convex hull of $\{ (x,g_1(x)) : x \in \{0,1\}^n \}$ is described by~\eqref{eq_no_good} and $0 \leq x_i \leq 1$ for $i=1,2,\dotsc,n$.
Such variables have applications in integer programming approaches for decoding of binary codes~\cite{ZhangS11}.
For now, the main purpose of~\eqref{eq_no_good} is to guarantee the \emph{existence} of integer programming formulations.
Once a more specific family of functions is identified to yield small linearizations for a certain application, improved integer programming relaxations can be developed for this family.
We will later see an example for this proposed approach.
We denote by $\B$ the set of all Boolean functions.
An interesting constrained subclass is the family $\C \subseteq \B$ of functions of the form
\begin{equation}
  g_{I,J}(x) \coloneqq \prod_{i \in I} x_i \cdot \prod_{j \in J} (1-x_j), \label{eq_complement_products}
\end{equation}
i.e., products of potentially complemented variables.
Even more restricted is the family $\M \subseteq \C$ of \emph{monomials}, i.e., functions $g_{I,\varnothing}$ for all $I$.

There often exist smaller and stronger formulations than~\eqref{eq_no_good}, e.g., if $g_i(x)$ does not depend on all $x$-variables.
For instance, for each function $g_{I,J} \in \C$, a perfect IP formulation is known, which is due to Fortet~\cite{Fortet60a,Fortet60b}, namely
\begin{subequations}
  \label{eq_complement_linearization}
  \begin{alignat}{7}
    y_{I,J} &\leq x_i &\quad& \forall i \in I, \\ 
    y_{I,J} &\leq 1-x_j &\quad& \forall j \in J, \\ 
    1-y_{I,J} &\leq \sum_{i \in I} (1-x_i) + \sum_{j \in J} x_j, \\
    y_{I,J} &\geq 0.
  \end{alignat}
\end{subequations}

On the one hand, the interplay of multiple simultaneous linearizations has been investigated by many researchers.
This includes mainly results about the quadratic case, e.g., for products of binary and continuous variables or products of binary variables and linear combinations of binary variables~\cite{Glover75}.
Worth mentioning is also the relaxation-linearization technique due to Adams and Sherali~\cite{AdamsS86,AdamsS90,AdamsS93}, which was also applied to polynomial optimization~\cite{SheraliT92}.
The strength of such alternative formulations for quadratic polynomials is compared theoretically and in practice in~\cite{FuriniT13}.
However, we are not aware of any research about their size.

On the other hand, the minimum number of additional \emph{quadratization variables} was investigated by Anthony et al.~\cite{AnthonyBCG17}.
These are variables $y \in \{0,1\}^k$ such that for all $x \in \{0,1\}^n$
\begin{equation}
  f(x) = \min \{ g(x,y) : y \in \{0,1\}^k \} \label{eq_quadratization}
\end{equation}
holds for a quadratic polynomial $g : \{0,1\}^{n + k} \to \R$.
This approach is not directly related to the linearization complexity, but similar properties of pseudo-Boolean functions are exploited.
Besides establishing first results on the linearization complexity with respect to different function families we showcase the use of a new alternative linearization technique for an application.
Most importantly, we present open problems and interesting research questions to stimulate further research in this direction.
Before presenting the outline of the paper, the potential of using various linearizations is exemplified.

\paragraph{Example.}
Consider $f(x) = x_1x_2 + x_1x_3 + x_2x_3 -x_1x_2x_3$.
Since it is the sum of four binary (non-affine) terms it has $\lc{\C}(f) \leq \lc{\M} = 4$.
However, it turns out that $\lc{\C}(f) \leq 1$ holds:
consider $g_1(x) = (1-x_1) (1-x_2) (1-x_3)$ and observe that $f(x) = x_1 + x_2 + x_3 - 1 + g_1(x)$, which yields $\lc{\C}(f) \leq 1$.
It is also not hard to see that $\lc{\C}(f) \geq 1$ holds (see \cref{thm_zero_affine}), which establishes $\lc{\C}(f) = 1$.

Clearly, this example can be extended to an arbitrary degree, showing that $\lc{\M}(f)$ can be exponentially larger (in $n$) than $\lc{\C}(f)$.

\paragraph{Outline.}
The paper is structured as follows.
In \cref{sec_basic_results} we present theoretical results about the linearization complexity for various families $\G$ of Boolean functions.
\Cref{sec_auto_correlation} is about the low auto-correlation binary sequences problem, which is an optimization problem that arises in theoretical physics, and for which we present a new linearization and evaluate it computationally.
We conclude the paper with several open problems in \cref{sec_open_problems} and hope that some of them will be addressed by researchers in the future.

\section{Basic results}
\label{sec_basic_results}

\paragraph{Nonlinear part of a function.}
We first introduce an auxiliary function with the purpose of removing the linear part of a given function.
For a function $f : \{0,1\}^n \to \R$ we denote by $\widetilde{f}$ the function defined by
\[
  \widetilde{f}(x) \coloneqq f(x) - f(\zerovec) - \sum_{i=1}^n \big( f(\unitvec{i}) - f(\zerovec) \big) x_i,
\]
which we call \emph{the nonlinear part of $f$}.
Here, $\zerovec$ and $\unitvec{i}$ denote the zero vector and the $i$-th standard unit vector, respectively.
The following proposition makes clear why the nonlinear portion is useful.

\begin{proposition}
  \label{thm_nonlinear_part}
  Let $f : \in \{0,1\}^n \to \R$.
  Then its nonlinear part satisfies the following properties.
  \begin{enumerate}
  \item
    $\widetilde{f}(\zerovec) = 0$ and $\widetilde{f}(\unitvec{i}) = 0$ holds for $i=1,2,\dotsc,n$.
  \item
    $\lc{\G}(\widetilde{f}) = \lc{\G}(f)$ holds for any family $\G$ of functions.
  \end{enumerate}
\end{proposition}

\begin{proof}
  To verify the first property it suffices to plug in the zero vector and the unit vectors.
  Now observe that
  $\widetilde{f}(x) = \transpose{a}x + \beta + f(x)$ where $a_i = f(\zerovec) - f(\unitvec{i})$ and $\beta = - f(\zerovec)$ holds.
  With this in mind, the second property follows readily from the definition of linearization complexity.
\end{proof}

The result implies that we only need to analyze linearization complexities of functions $f$ with $f(\zerovec) = f(\unitvec{i}) = 0$ for $i=1,2,\dotsc,n$.

\paragraph{Linearization complexity.}
We continue with simple properties of the linearization complexity.
Our first observation is that for arbitrary $\G \subseteq \B$, $\lc{\G}(f) > 0$ indicates that $f$ is actually nonlinear over $\{0,1\}^n$.
\begin{proposition}
  \label{thm_zero_affine}
  For $\G \subseteq \B$, a function $f$ has $\lc{\G}(f) = 0$ if and only if $f$ is affine.
\end{proposition}
We can rephrase the result in terms of the nonlinear part of $f$ as follows. It holds $\lc{\G}(\widetilde{f}) = 0$ if and only if $\widetilde{f}$ is the zero function.

The second property is the monotonicity of $\lc{\G}(f)$ with respect to the family $\G$, which also follows from the definition.
\begin{proposition}
  \label{thm_monotone}
  Let $\G' \subseteq \G$ and consider $f : \{0,1\}^n \to \R$.
  Then $\lc{\G'}(f) \geq \lc{\G}(f)$.
\end{proposition}

The third property is a characterization the linearization complexity for $\G = \M$ by means of a polynomial representation of $f$.
In fact, the following result follows readily from the well-known fact that every pseudo-Boolean function has a unique multilinear polynomial representation~\cite{HammerRR63,HammerR12} (see also Proposition~2 in \cite{BorosH02}).
It also yields finiteness of the linearization complexity for all $\G \supseteq \M$.

\begin{proposition}
  \label{thm_lc_monomials}
  Let $f : \{0,1\}^n \to \R$.
  Then $\lc{\M}(f)$ is equal to the number of monomials (with a nonzero coefficient) of degree at least $2$ of the (unique) multilinear polynomial $p$ with $p(x) = f(x)$ for all $x \in \{0,1\}^n$.
  In particular, $\lc{\M}(f) \leq 2^n - n - 1$ holds.
\end{proposition}

\begin{proof}
  For a function $f : \{0,1\}^n \to \R$ the mentioned polynomial can be defined as
  \begin{equation*}
    p(x) \coloneqq \sum_{X \subseteq \{1,2,\dots,n\}} f(\chi(X)) \cdot \prod_{i \in X} x_i \cdot \prod_{i \notin X} (1-x_i),
  \end{equation*}
  where $\chi(X) \in \{0,1\}^n$ denotes the characteristic vector of $X$, defined via $\chi(X)_i = 1 \iff i \in X$.
  Note that $p$ has degree at most $n$ and hence is the weighted sum of $\transpose{a}x + \beta$ (for suitable $a \in \R^n$ and $\beta \in \R$) and at most $2^n-n-1$ monomials in $\M$ of degree greater than $1$.
  By construction, we have $f(x) = p(x)$ for each $x \in \{0,1\}^n$.
  This establishes the upper bound on $\lc{\M}(f)$, and equality follows from the uniqueness of $p$'s representation.
\end{proof}

\Cref{thm_lc_monomials} yields a sufficient condition, namely $\G \supseteq \M$, for $\lc{\G}(f)$ to be finite for all functions $f$.
The following proposition provides a characterization.
\begin{proposition}
  \label{thm_lc_finite}
  Let $\G_n$ be a subset of functions that map from $\{0,1\}^n \to \R$.
  Then $\lc{\G_n}(f) < \infty$ holds for all $f : \{0,1\}^n \to \R$ if and only if $\lc{\G_n}(g_{I,\varnothing}) < \infty$ holds for all $I \subseteq \{1,2,\dotsc,n\}$.
\end{proposition}
\begin{proof}
  Necessity is trivial.
  Sufficiency follows by expressing $f$ as a weighted sum of monomials and combining the monomials' linearizations accordingly.
\end{proof}

\paragraph{Random polynomials.}
Our first larger result essentially states that having a small linearization complexity is an exceptional property in a probabilistic sense.

\begin{theorem}
  \label{thm_null_set}
  Consider a family $\G$ of functions with $\M \subseteq \G \subseteq \B$.
  Then the subset of functions $f : \{0,1\}^n \to \R$ with $\lc{\G}(f) < 2^n - n - 1$ is a null set (in the set of pseudo-Boolean functions $f : \{0,1\}^n \to \R$).
\end{theorem}

The statement and its proof are similar to Theorem~1 in \cite{AnthonyBCG17}, basically exploiting that pseudo-Boolean functions form a vector space of dimension $2^n$ and that any subspace spanned by fewer functions constitutes a null set.

\begin{proof}
  Consider a fixed number $n$ of variables and the subset $\G_n \subseteq \G$ of functions that map from $\{0,1\}^n$.
  We consider the system of equations
  \begin{equation}
    \transpose{\bar{x}} a + \beta + \sum_{g \in \G_n :\, g(\bar{x}) = 1} c_g = f(\bar{x}) \qquad \forall \bar{x} \in \{0,1\}^n \label{eq_system}
  \end{equation}
  in variables $a_i \in \R$ (for $i = 1,2,\dotsc,n$), $\beta \in \R$ and $c_g \in \R$ for all $g \in \G_n$.
  On the one hand, every solution $(a,\beta,c)$ yields a linearization by letting $g_1,g_2,\dotsc,g_k$ be those $g \in \G_n$ for which $c_g \neq 0$ and letting $b_i \coloneqq c_{g_i}$ for $i=1,2,\dotsc,k$.
  Note that its size $k$ is equal to the number of nonzero components of the $c$-vector.
  On the other hand, every linearization, say with $a \in \R^n$, $\beta \in \R$, $b \in \R^k$ and $g_1,g_2,\dotsc,g_k \in \G_n$, corresponds to a solution $(a,\beta,c)$, where $c_{g_i} \coloneqq b_i$ for $i=1,2,\dotsc,k$ and $c_g \coloneqq 0$ for all $g \in \G_n \setminus \{ g_1,g_2,\dotsc,g_k \}$.

  Hence, a function $f : \{0,1\}^n \to \R$ has $\lc{\G}(f) < 2^n - n - 1$ if and only if its right-hand side vector $f(\bar{x})_{\bar{x} \in \{0,1\}} \in \R^{2^n}$ in the system~\eqref{eq_system} lies in the span of the first $n+1$ columns the coefficient matrix of~\eqref{eq_system} and less than $2^n-n-1$ of the remaining columns.
  However, the set of such right-hand side vectors (for a fixed choice of columns) has dimension less than $2^n$ and thus constitutes a null set.
  This fact is not changed by considering any of the (finite) number of column subsets (consisting of less than $2^n$ columns).
  The one-to-one correspondence between functions $f : \{0,1\}^n$ and their right-hand sides~\eqref{eq_system} establishes the result.
\end{proof}

\begin{corollary}
  \label{thm_random_polynomial}
  Consider a family $\G$ of functions with $\M \subseteq \G \subseteq \B$.
  A polynomial $p : \{0,1\}^n \to \R$ whose coefficients are chosen randomly according to any absolutely continuous probability distribution has $\lc{\G}(f) = 2^n - n - 1$ with probability~$1$.
\end{corollary}

\DeclareDocumentCommand\partialSumSet{}{\mathop{pss}}

\paragraph{Arbitrary functions.}
We already observed that $\lc{\B}(f) \leq 2^n - n - 1$ holds for any function $f : \{0,1\}^n \to \R$.
It turns out that the family $\B$ is so general that the corresponding linearization complexity only depends on the range of $\widetilde{f}$.
To make this precise, we define for a vector $w \in \R^k$ its \emph{partial sum set $\partialSumSet(w)$} as the set of partial sums $\sum_{i \in I} w_i$ over all subsets $I \subseteq \{1,2,\dotsc,k\}$.
We also allow for $k = 0$ and define $\partialSumSet(w) = \{ 0 \}$ in this case.

\begin{theorem}
  \label{thm_arbitrary}
  Let $f : \{0,1\}^n \to \R$ and let $Y \coloneqq \{ \widetilde{f}(x) : x \in \{0,1\}^n \}$ be the range of its nonlinear part.
  Then $\lc{\B}(f)$ is equal to the smallest dimension $k$ of a vector $w \in \R^k$ with $\partialSumSet(w) \supseteq Y$.
\end{theorem}

\begin{proof}
  Let $f$ and $Y$ be as in the theorem and consider a vector $w \in \R^k$ (for some $k$) with $\partialSumSet(w) \supseteq Y$.
  We now construct a linearization of $f$ of size $k$.
  For each $x \in \{0,1\}^n$ we have $\widetilde{f}(x) \in Y \subseteq \partialSumSet(w)$, which implies that there must be a subset $I_x \subseteq \{1,2,\dotsc,k\}$ with $\sum_{i \in I_x} w_i = \widetilde{f}(x)$.
  We now define, for $i=1,2,\dotsc,k$, the function $g_i : \{0,1\}^n \to \{0,1\}$ such that
  \[
    g_i(x) = 1 \iff i \in I_x
  \]
  holds.
  Moreover, we use $b \coloneqq w$ and observe
  $
    \sum_{i=1}^k b_i g_i(x)
    = \sum_{i \in I_x} w_i
    = \widetilde{f}(x)
  $
  for each $x \in \{0,1\}^n$, which establishes $\lc{\B}(\widetilde{f}) \leq k$ and by \cref{thm_nonlinear_part} also $\lc{\B}(f) \leq k$.

  For the reverse direction, suppose $\lc{\B}(f) = k$.
  Again by \cref{thm_nonlinear_part} there exist vectors $a \in \R^n$, $b \in \R^k$, scalar $\beta \in \R$ and functions $g_1,g_2, \dotsc, g_k : \{0,1\}^n \to \{0,1\}$ such that
  \[
    f(x) = \transpose{a}x + \beta + \sum_{i=1}^k b_i g_i(x)
  \]
  holds for all $x \in \{0,1\}^n$.
  For the nonlinear part we obtain
  \[
    \widetilde{f}(x) = \sum_{i=1}^k b_i \widetilde{g}_i(x).
  \]
  Now consider a $y \in Y$, i.e., the function value $y = \widetilde{f}(x)$ for some $x \in \{0,1\}^n$.
  By using the set $I \coloneqq \{ i \in \{1,2,\dotsc,k\} : \widetilde{g}_i(x) = 1 \}$ we obtain that $y \in \partialSumSet(b)$ holds.
  This implies $\partialSumSet(b) \supseteq Y$, and hence the choice $w \coloneqq b$ concludes the proof.
\end{proof}

The intuition behind \cref{thm_arbitrary} is the following: since there are no restrictions on the structure of the functions $g$ used, the fact that the possible arguments $x$ of $f$ are actually (binary) vectors is not relevant, i.e., one can think of just $2^n$ different inputs without any structure.
In the next section we will use this viewpoint to construct a small linearization for an example application.

\section{Application to Low Autocorrelation Binary Sequences}
\label{sec_auto_correlation}

In this section we consider the low auto-correlation binary sequences problem, which arises in theoretical physics when studying ground states of the Bernasconi model~\cite{Bernasconi87,Golay82}.
An extension of the problem was considered in~\cite{LiersMPRS10}, which involves, in addition to the dimension $N$, another parameter $R \in \{1,2,\dotsc,N\}$ indicating the so-called \emph{interaction range} to specify a problem instance.
In the Bernasconi model, $R = N$ holds.
The general formulation is the following, where we omit normalization factors.
\begin{subequations}
  \label{model_bernasconi}
  \begin{alignat}{7}
    & \text{min } ~\mathrlap{ \sum_{i=0}^{N-R} \sum_{d=1}^{R-1} \left( \sum_{j=i+1}^{i+R-d} s_j s_{j+d} \right)^2 } \label{model_bernasconi_obj} \\
    & \text{s.t. }
      & s_i \in \{-1,+1\} &\quad& i = 1,2,\dotsc,N \label{model_bernasconi_domain}
  \end{alignat}
\end{subequations}
The state-of-the-art method for solving the Bernasconi problem ($R = N$) is a highly parallelized combinatorial branch-and-bound algorithm~\cite{PackebuschM16} which builds upon earlier work~\cite{Mertens96,MertensB98}.
The optima for the problem are known for all $N \leq 66$.
A dynamic programming algorithm for unconstrained binary polynomial optimization was recently proposed and tested computationally~\cite{ClausenCLRR24}.
Moreover, also methods based on quadratic convex reformulations were applied successfully to the problem~\cite{BillionnetE07,ElloumiLL21}.
For the Bernasconi model, they can find optimal solutions up to $N = R = 30$, but perform better for smaller interaction range, e.g., obtaining an optimal solution for $(N,R) = (40,10)$.

In this section we investigate IP models based on binary linearizations for this problem.
To this end, let us first rephrase the problem by substituting $s_i \in \{-1,+1\}$ with $2x_i - 1$ for $x_i \in \{0,1\}$ for all $i \in \{1,2,\dotsc,N\}$.
We obtain
\begin{subequations}
  \label{minlp_bernasconi}
  \begin{alignat}{7}
    & \text{min } ~\mathrlap{ f^{\text{bern}}_{N,R}(x) \coloneqq \sum_{i=0}^{N-R} \sum_{d=1}^{R-1} \left( \sum_{j=i+1}^{i+R-d} (4x_j x_{j+d} - 2x_j - 2x_{j+d} + 1) \right)^2 } \label{minlp_bernasconi_obj} \\
    & \text{s.t. }
      & x_i \in \{0,1\} &\quad& i = 1,2,\dotsc,N. \qquad\qquad\qquad\qquad\qquad\qquad\qquad \label{minlp_bernasconi_domain}
  \end{alignat}
\end{subequations}
Let $f^{\text{bern}}_{N,R}(x) = \sum_{m=1}^t a_m g_{I_m,\varnothing}(x)$ be the decomposition of $f^{\text{bern}}_{N,R}$ into its $t$ monomials.
With this notation, the standard IP reformulation~\eqref{eq_complement_linearization} applied to $f^{\text{bern}}_{N,R}$ reads as follows.
\begin{subequations}
  \label{ip_bernasconi_standard}
  \begin{alignat}{7}
    & \text{min } ~\mathrlap{ \sum_{m=1}^t a_m z_{I_m} } \\
    & \text{s.t. }
      & z_{I_m} &\leq x_i &\quad& \forall i \in I_m,~ m =1,2,\dotsc,t \\
    & & 1-z_{I_m} &\leq \sum_{i \in I_m} (1-x_i) &\quad& m=1,2,\dotsc,t \\
    & & z_{I_m} &\in \{0,1\} &\quad& m=1,2,\dotsc,t \\
    & & x_i  &\in \{0,1\} &\quad& i=1,2,\dotsc,N
  \end{alignat}
\end{subequations}
The following proposition establishes the asymptotic number of monomials of $f^{\text{bern}}_{N,R}(x)$, which implies that~\eqref{ip_bernasconi_standard} has $\orderTheta(N^3)$ many variables and constraints.

\begin{proposition}
  \label{thm_labs_monomials_cubic}
  For $N \geq 3$ and $1 \leq R \leq N$, the number of monomials of $f^{\text{bern}}_{N,R}(x)$ and thus the number of variables and constraints of model~\eqref{ip_bernasconi_standard} is $\orderTheta(NR^2)$.
\end{proposition}

\begin{proof}
  We first establish the upper bound.
  For that it suffices to consider the polynomial~\eqref{model_bernasconi_obj} in the $s$-variables since substitution of $s_i$ by $2x_i-1$ yields, due to the degree being at most $4$, a blow-up factor of $16$ in the number of monomials.
  All monomials with a nonzero coefficient are of the form $s_j s_{j+d} s_{j'} s_{j'+d}$.
  Since there are only $\orderO(N)$ choices for $j$ and, after chosing $j$, only $\orderO(R)$ choices for $j'$ and for $d$, there are only $\orderO(NR^2)$ monomials with a non-zero coefficient.

  It remains to prove the lower bound for which we analyze the polynomial~\eqref{minlp_bernasconi_obj} in the $x$-variables to avoid a discussion of potential cancellations caused by the substitution.
  To this end, let $k \coloneqq \lfloor R/4 \rfloor$.
  We consider any $d \in \{1,2,\dotsc,k\}$, any $j' \in \{k+1,k+2,\dotsc,2k\}$ and any $i \in \{1,\dotsc,\lfloor N/4 \rfloor\}$.
  Note that for each such $(d,j',i)$, the term $x_i x_{i+d} x_{i+j'} x_{i+j'+d}$ arises as a monomial of $f^{\text{bern}}_{N,R}$.
  Moreover, all these monomials are distinct and only appear with positive coefficients.
  Hence, due to $k \in \orderOmega(R)$ there exist $\orderTheta(N R^2)$ such monomials with a positive coefficient.
\end{proof}

\paragraph{Partial sum sets.}
With \cref{thm_labs_monomials_cubic} in mind a natural follow-up question is whether there exist formulations of smaller asymptotic size.
By \cref{thm_arbitrary} we can investigate the partial sum set of the domain $Y$ of its nonlinear part.
The set $Y$ is nontrivial to characterize, in particular since computing $\min(Y)$ is equivalent to solving~\eqref{model_bernasconi}.
Nevertheless, we can construct a suitable partial sum set of quadratic size which yields our main theoretical result for the application.

\begin{theorem}
  \label{thm_bernasconi_lc}
  For $N \geq 3$ and $1 \leq R \leq N$ we have $\lc{\B}(f^{\text{bern}}_{N,R}) \leq (N-R+1) R^2$.
\end{theorem}

\begin{proof}
  First observe that the sum $\sum_{j=i+1}^{i+R-d} s_j s_{j+d}$ can attain $R-d+1$ different values, and these depend on $d$ (in addition to $R$).
  Consequently there are also only $R-d+1$ different squared such values.
  Let, for $i \in \{0,1,\dotsc,N-R\}$ and $d \in \{1,2,\dotsc,R-1 \}$, the vector $w^{i,d}$ consist of these squared values in an arbitrary order and let $w = (w^{0,1}, w^{0,2}, \dotsc, w^{0,R-1},w^{1,1},w^{1,2}, \dotsc, w^{N-R,R-1})$ be the concatenation of all these vectors.
  Hence, $w$ has dimension at most $(N-R+1) (R-1) R \leq (N-R+1) R^2$.
  Moreover, for each $x \in \{0,1\}^n$, the value $f^{\text{bern}}_{N,R}(x)$ is a partial sum of elements of $w$.
  The result now follows from \cref{thm_arbitrary}.
\end{proof}

\DeclareDocumentCommand\domainID{}{\mathcal{ID}}

Note that the theorem only yields an asymptotic improvement over \cref{thm_labs_monomials_cubic} if $R$ is close to $N$.
In particular for the Bernasconi model, we improve from $\orderTheta(N^3)$ to $N^2$.

\paragraph{Model with value indicators.}
\Cref{thm_bernasconi_lc} only yields the existence of an IP formulation with $(N-R+1)R^2$ auxiliary variables.
Instead of following its proof to derive such a formulation we directly exploit the main observation that led to the quadratic size, namely that every element from the domain is a sum of squares and that, for every $(i,d) \in \domainID \coloneqq \{ 0,1,\dotsc,N-R \} \times \{ 0,1,\dotsc,R-1 \}$, at most $R-d+1$ such squared numbers can actually arise.
This yields the idea of introducing binary variables for indicating (for each pair $(i,d) \in \domainID$ separately) which number was actually squared.
To this end, we define the set $L_{i,d} \coloneqq \left\{ \sum\limits_{j=i+1}^{i+R-d} s_j s_{j+d} : s \in \{-1,1\}^N \right\}$ of distinct function values of the inner sum of~\eqref{model_bernasconi_obj}.
Since $s_j \cdot s_{j+d} \in \{-1,+1\}$ holds, we have
\[
  L_{i,d} \subseteq \{ -(R-d), -(R-d) + 2, \dotsc, (R-d)-2, (R-d) \}.
\]
We obtain $|L_{i,d}| \leq R+1-d$ for all $(i,d) \in \domainID$.
Now we consider, for each such pair $(i,d)$ and each $\ell \in L_{i,d}$, the function $g_{i,d,\ell} : \{0,1\}^N \to \{0,1\}$ that indicates whether the inner sum of~\eqref{model_bernasconi_obj} for some $(i,d)$ is equal to $\ell$.
More precisely, it shall satisfy $g_{i,d,\ell}(x) = 1$ if and only if $\sum\limits_{j=i+1}^{i+R-d} (2x_j-1)(2x_{j+d}-1) = \ell$.
Since for each $(i,d)$ and each $x \in \{0,1\}^n$, precisely one of the functions $g_{i,d,\ell}(x)$ is equal to $1$, we obtain
\[
  f^{\text{bern}}_{N,R}(x) = \sum_{(i,d) \in \domainID} \sum_{\ell \in L_{i,d}} \ell^2 g_{i,d,\ell}(x).
\]
This clearly leads to the following \emph{value-indicator-nogood (VING)} formulation with only $\orderO((N-R) R^2)$ variables by enforcing $z_{i,d,\ell} = g_{i,d,\ell}(x)$ via constraints~\eqref{eq_no_good}.
However, these are exponentially many constraints.
\begin{subequations}
  \label{ip_only_indicators}
  \begin{alignat}{7}
    & \text{min } ~\mathrlap{ \sum_{(i,d) \in \domainID} \sum_{\ell \in L_{i,d}} \ell^2 z_{i,d,\ell} } \label{ip_only_indicators_obj} \\
    & \text{s.t. }
      & \sum_{j : \bar{x}_j = 0} x_j &\;+\;& \sum_{j : \bar{x}_j = 1} (1-x_j) &\;+\; z_{i,d,\ell} &\;\geq 1 &\quad \forall (i,d) \in \domainID,~ \forall \ell \in L_{i,d}, \nonumber \\[-2ex]
    & & & & & & &\quad \qquad \forall \bar{x} \; \in \{0,1\}^n : g_{i,d,\ell}(\bar{x}) = 1 \label{ip_only_indicators_nogood1} \\[1ex]
    & & \sum_{j : \bar{x}_j = 0} x_j &\;+\;& \sum_{j : \bar{x}_j = 1} (1-x_j) &\;+\; (1-z_{i,d,\ell}) &\;\geq 1 &\quad \forall (i,d) \in \domainID,~ \forall \ell \in L_{i,d}, \nonumber \\[-2ex]
    & & & & & & &\quad \qquad \forall \bar{x} \; \in \{0,1\}^n : g_{i,d,\ell}(\bar{x}) = 0 \label{ip_only_indicators_nogood2} \\[1ex]
    & & & & x_i &\in \{0,1\} & &\quad \forall i \in \{1,2,\dotsc,N\} \label{ip_only_indicators_domain_x}\\
    & & & & z_{i,d,\ell} &\in \{0,1\} & &\quad \forall (i,d) \in \domainID,~ \forall \ell \in L_{i,d} \label{ip_only_indicators_domain_z}
  \end{alignat}
\end{subequations}

We can build upon this formulation idea in order to devise a hybrid formulation, called \emph{value-indicator-quadratic (VIQ)}:
Instead of relating the $z$-variables directly to the $x$-variables we introduce variables $y_{i,j}$ that indicate for all $i,j$ (with $i \neq j$) whether the product $s_i s_j$ is $+1$ or $-1$, and link these to the $z$-variables.

\begin{subequations}
  \label{ip_indicators}
  \begin{alignat}{7}
    & \text{min } ~\mathrlap{\sum_{(i,d) \in \domainID} \sum_{\ell \in L_{i,d}} \ell^2 z_{i,d,\ell} } \label{ip_indicators_obj} \\
    & \text{s.t. }
      & x_i + x_j &\geq 1 - y_{i,j} &\quad& \forall i,j \in \{1,2,\dotsc,N\} : i < j < i + R \label{ip_indicators_quad1} \\
    & & x_i - x_j &\geq y_{i,j} - 1 &\quad& \forall i,j \in \{1,2,\dotsc,N\} : i < j < i + R \label{ip_indicators_quad2}\\
    & & -x_i + x_j &\geq y_{i,j} - 1 &\quad& \forall i,j \in \{1,2,\dotsc,N\} : i < j < i + R \label{ip_indicators_quad3}\\
    & & x_i + x_j &\leq y_{i,j} + 1 &\quad& \forall i,j \in \{1,2,\dotsc,N\} : i < j < i + R \label{ip_indicators_quad4}\\
    & & \sum_{\ell \in L_{i,d}} z_{i,d,\ell} &=1 &\quad& \forall (i,d) \in \domainID \label{ip_indicators_convexity} \\
    & & \sum_{j=i+1}^{i+R-d} (2y_{j,j+d}-1) &=\sum_{\ell \in L_{i,d}} \ell z_{i,d,\ell} &\quad& \forall (i,d) \in \domainID \label{ip_indicators_indicate} \\
    & & x_i &\in \{0,1\} &\quad& i=1,2,\dotsc,N \label{ip_indicators_domain_x}\\
    & & y_{i,j} &\in \{0,1\} &\quad& \forall i,j \in \{1,2,\dotsc,N\} : i < j < i + R \label{ip_indicators_domain_y}\\
    & & z_{i,d,\ell} &\in \{0,1\} &\quad& \forall (i,d) \in \domainID,~ \forall \ell \in L_{i,d} \label{ip_indicators_domain_z}
  \end{alignat}
\end{subequations}

\begin{corollary}
  \label{thm_indicators_ip}
  For $N \geq 3$ and $1 \leq R \leq N$, the integer program~\eqref{ip_indicators} correctly models the low auto-correlation binary sequences problem~\eqref{model_bernasconi} and has $\orderO((N-R)R^2)$ variables and constraints.
\end{corollary}

\begin{proof}
  Similar to~\eqref{minlp_bernasconi} we model $s_i \in \{-1,+1\}$ by $s_i = 2x_i - 1$, i.e., $x_i = 1$ if and only if $s_i = 1$.
  Constraints~\eqref{ip_indicators_quad1}--\eqref{ip_indicators_quad4} enforce that $2y_{i,j} - 1 = s_i \cdot s_j = (2x_i-1)(2x_j-1)$ holds, i.e., $y_{i,j} = 1$ if and only if $x_i = x_j$ holds.
  For each $d \in \{1,2,\dotsc,N-1\}$, equation~\eqref{ip_indicators_convexity} implies that $z_{i,d,\ell} = 1$ holds for exactly one $\ell \in L_{i,d}$.
  Since the left-hand side of~\eqref{ip_indicators_indicate} is equal to the inner sum of~\eqref{model_bernasconi_obj}, the right-hand side implies that $z_{i,d,\ell} = 1$ holds exactly for $\ell$ being the value of that inner sum.
  Hence, the right-hand side of~\eqref{ip_indicators_indicate} is the corresponding $\ell \in L_{i,d}$.
  This implies that the contribution of all $z_{i,d,\ell}$ for some $(i,d) \in \domainID$ to the objective~\eqref{ip_indicators_obj} is equal to $\ell^2$, where $\ell$ is the value of the inner sum of~\eqref{model_bernasconi_obj}.
  We conclude that~\eqref{model_bernasconi_obj} is indeed equal to $f^{\text{bern}}_{N,R}(x)$.

  The numbers of variables and constraints are easily verified.
\end{proof}

\paragraph{Experimental evaluation.}
The numbers of variables and constraints are only one factor that has impact on the solution time of an IP.
For instance, the quality of the dual (in our case lower) bounds on the optimum and the ability to find good primal solutions is extremely important.
Hence, we compare the new IP models~\eqref{ip_indicators} and~\eqref{ip_only_indicators} with the standard formulation~\eqref{ip_bernasconi_standard}.

We tried to solve the low auto-correlation binary sequences problem with both models using the SCIP solver framework~\cite{SCIP9}\footnote{The code can be found at \href{https://github.com/discopt/labs}{github.com/discopt/labs}.}.
Constraints~\eqref{ip_only_indicators_nogood1} and~\eqref{ip_only_indicators_nogood2} were implemented as a constraint handler that finds violated inequalities in linear time.
Note that all three approaches are too slow to compete with the most recent state-of-the-art approach that is based on combinatorial branch-and-bound with which the problem could be solved to optimality up to $N = 66$~\cite{PackebuschM16}.
We report about instances with $N \in \{5,10,15,20,25,30,35\}$ and $R \in \{N, \lfloor \tfrac{3}{4}N \rceil, \lfloor \tfrac{1}{2}N \rceil, \lfloor \tfrac{1}{4}N \rceil, \lfloor \tfrac{1}{8}N \rceil \}$ using version~9.0.0 of SCIP using SoPlex~7.0.0 as an LP solver.
% 248 cores * 55d * 24h  = 327360
Also note that we ran our experiments only on a single core (with an Intel CPU on \SI{2.1}{\giga\hertz}) while the computations for $N=66$ were done using 248 cores and took about 55 days~\cite{PackebuschM16}.
Moreover, we strengthened none of the models to keep the comparison fair.

\Cref{table_bernasconi_ip_sizes} provides an overview over the number of variables and constraints for the models, while \cref{table_bernasconi_ip_solved} shows the actual results.
Note that for $R = N \geq 25$, SCIP could not solve the standard formulation~\eqref{ip_bernasconi_standard} within three hours.
In fact, the lower bound obtained after that time was still a negative number.
The first model~\eqref{ip_only_indicators} with value indicators always produces nonnegative bounds by construction, but it usually takes a lot of branching effort until this increases.
A likely reason is that, besides providing a nonnegative bound, the LP relaxation is not very tight, and hence the solver has to effectively solve the problem by branching.
The same holds for the hybrid model~\eqref{ip_indicators} but, as can be seen in \cref{table_bernasconi_ip_solved} it requires much less branching.
We conclude that our last formulation~\eqref{ip_indicators} outperforms the standard formulation as well as the straight-forward formulation~\eqref{ip_only_indicators}, in particular by means of quality of the bound obtained from the LP relaxation.

Finally, let us remark that neither model outperforms the mixed-integer quadratic programming approach described in~\cite{ElloumiLL21} or the dynamic programming algorithm from~\cite{ClausenCLRR24}, although a direct comparison of the bounds with those from several other papers is difficult due to the fact that the benchmark instance files in~\cite{MINLPLib20} used there do not contain the correct absolute term in the polynomial.
However, we believe that it is possible to make formulation~\eqref{ip_indicators} competitive by strengthening the model using cutting planes that improve the linking of $y$- and $z$-variables.

\DeclareDocumentCommand\timelimit{}{\faClockO}

\begin{table}[htpb]
  \caption{%
    Comparison of IP models~\eqref{ip_bernasconi_standard}, \eqref{ip_only_indicators} and~\eqref{ip_indicators} with respect to model size.
    The first two columns show the parameters $N$ and $R$.
    Moreover, for each model, the number of variables (Vars) and constraints (Cons) are shown, where for~\eqref{ip_only_indicators} the number indicates the number of constraints that were generated while solving the IP with a time limit of \SI{3}{\hour} (if reached, this is indicated with ``$>$'').
  }
  \label{table_bernasconi_ip_sizes}
  \setlength{\tabcolsep}{1.5mm}
  \begin{center}
    \begin{tabular}{rr|rr|rr|rr|rr}
      \multicolumn{2}{c|}{\textbf{Instance}}
      & \multicolumn{2}{c|}{\textbf{Standard~(\ref{ip_bernasconi_standard})}}
      & \multicolumn{2}{c|}{\textbf{VING~(\ref{ip_only_indicators})}}
      & \multicolumn{2}{c}{\textbf{VIQ~(\ref{ip_indicators})}} \\
      \textbf{N} & \textbf{R}
      & \textbf{Vars} & \textbf{Cons}
      & \textbf{Vars} & \textbf{Cons} 
      & \textbf{Vars} & \textbf{Cons} \\
      \hline
      5 & 3  & \num{14} & \num{20}  & \num{20} & \num{87}   & \num{30} & \num{52} \\
      5 & 4  & \num{29} & \num{76}  & \num{23} & \num{198}   & \num{33} & \num{52} \\
      5 & 5  & \num{33} & \num{92}  & \num{19} & \num{200}   & \num{29} & \num{48} \\
      10 & 3  & \num{29} & \num{45}  & \num{50} & \num{519}   & \num{95} & \num{212} \\
      10 & 5  & \num{98} & \num{312}  & \num{94} & \num{7441}   & \num{139} & \num{228} \\
      10 & 8  & \num{197} & \num{723}  & \num{115} & \num{13902}   & \num{160} & \num{222} \\
      10 & 10  & \num{232} & \num{870}  & \num{64} & \num{12944}   & \num{109} & \num{198} \\
      15 & 2  & \num{16} & \num{1}  & \num{43} & \num{634}   & \num{148} & \num{448} \\
      15 & 4  & \num{119} & \num{356}  & \num{123} & \num{22808}   & \num{228} & \num{492} \\
      15 & 8  & \num{387} & \num{1463}  & \num{295} & \num{583033}   & \num{400} & \num{532} \\
      15 & 11  & \num{609} & \num{2398}  & \num{340} & \num{644103}   & \num{445} & \num{520} \\
      15 & 15  & \num{784} & \num{3141}  & \num{134} & \num{642542}   & \num{239} & \num{448} \\
      20 & 3  & \num{59} & \num{95}  & \num{110} & \num{8174}   & \num{300} & \num{832} \\
      20 & 5  & \num{228} & \num{752}  & \num{244} & \num{1276533}   & \num{434} & \num{888} \\
      20 & 10  & \num{854} & \num{3368}  & \num{614} & \num{18843315}   & \num{804} & \num{958} \\
      20 & 15  & \num{1515} & \num{6175}  & \num{734} & \num{27263821}   & \num{924} & \num{928} \\
      20 & 20  & \num{1880} & \num{7732}  & \num{229} & \num{27916846}   & \num{419} & \num{798} \\
      25 & 3  & \num{74} & \num{120}  & \num{140} & \num{21822}   & \num{440} & \num{1292} \\
      25 & 6  & \num{433} & \num{1552}  & \num{425} & \num{56896414}   & \num{725} & \num{1400} \\
      25 & 13  & \num{1808} & \num{7342}  & \num{1195} & \num{84930754}   & \num{1495} & \num{1512} \\
      25 & 19  & \num{3066} & \num{12708}  & \num{1348} & \num{93968461}   & \num{1648} & \num{1452} \\
      25 & 25  & \num{3703} & \num{15437}  & \num{349} & \num{402081122}   & \num{649} & \num{1248} \\
      30 & 4  & \num{254} & \num{776}  & \num{273} & \num{9603173}   & \num{708} & \num{1902} \\
      30 & 8  & \num{957} & \num{3683}  & \num{835} & \num{45645241}   & \num{1270} & \num{2062} \\
      30 & 15  & \num{2975} & \num{12235}  & \num{1934} & \num{116846391}   & \num{2369} & \num{2188} \\
      30 & 23  & \num{5407} & \num{22633}  & \num{2230} & \num{169575296}   & \num{2665} & \num{2092} \\
      30 & 30  & \num{6443} & \num{27079}  & \num{494} & \num{346309883}   & \num{929} & \num{1798} \\
      35 & 4  & \num{299} & \num{916}  & \num{323} & \num{46091169}   & \num{918} & \num{2572} \\
      35 & 9  & \num{1417} & \num{5557}  & \num{1223} & \num{66844682}   & \num{1818} & \num{2812} \\
      35 & 18  & \num{5038} & \num{20957}  & \num{3095} & \num{121894574}   & \num{3690} & \num{2992} \\
      35 & 26  & \num{8383} & \num{35287}  & \num{3535} & \num{144974813}   & \num{4130} & \num{2880} \\
      35 & 35  & \num{10288} & \num{43472}  & \num{664} & \num{327766744}   & \num{1259} & \num{2448}
    \end{tabular}
  \end{center}
\end{table}

\begin{landscape}

\begin{table}[htpb]
  \caption{%
    Comparison of IP models~\eqref{ip_bernasconi_standard}, \eqref{ip_only_indicators} and~\eqref{ip_indicators}.
    The first three columns show the parameters $N$ and $R$ as well as the best-known solution value.
    Moreover, for each model, the dual bound from the LP relaxation (LP), the dual bound after termination of branch-and-cut (IP), the number of solved branch-and-bound nodes (Nodes), as well as the total solution time (Time) are reported.
    The symbol \faClockO{} indicates that the computation reached the timeout of \SI{3}{\hour}.
  }
  \label{table_bernasconi_ip_solved}
  \setlength{\tabcolsep}{1.5mm}
  \begin{center}
    \begin{tabular}{rrr|rrrr|rrrr|rrrr}
      \multicolumn{3}{c|}{\textbf{Instance}}
      & \multicolumn{4}{c|}{\textbf{Standard~(\ref{ip_bernasconi_standard})}}
      & \multicolumn{4}{c|}{\textbf{VING~(\ref{ip_only_indicators})}}
      & \multicolumn{4}{c}{\textbf{VIQ~(\ref{ip_indicators})}} \\
      \textbf{N} & \textbf{R} & \textbf{Best}
      & \textbf{LP} & \textbf{IP} & \textbf{Nodes} & \textbf{Time}
      & \textbf{LP} & \textbf{IP} & \textbf{Nodes} & \textbf{Time}
      & \textbf{LP} & \textbf{IP} & \textbf{Nodes} & \textbf{Time} \\
      \hline
      5 & 3 & \num{3}  & \num{3} & \num{3} & \num{1} & \SI{0.8}{\second}   & \num{0} & \num{3} & \num{7} & \SI{0.4}{\second}   & \num{3} & \num{3} & \num{1} & \SI{0.4}{\second} \\
      5 & 4 & \num{4}  & \num{-148} & \num{4} & \num{1} & \SI{1.0}{\second}   & \num{0} & \num{4} & \num{15} & \SI{0.1}{\second}   & \num{4} & \num{4} & \num{1} & \SI{0.0}{\second} \\
      5 & 5 & \num{2}  & \num{-226} & \num{2} & \num{1} & \SI{1.1}{\second}   & \num{0} & \num{2} & \num{19} & \SI{0.1}{\second}   & \num{2} & \num{2} & \num{1} & \SI{0.0}{\second} \\
      10 & 3 & \num{8}  & \num{8} & \num{8} & \num{1} & \SI{1.0}{\second}   & \num{0} & \num{8} & \num{33} & \SI{0.0}{\second}   & \num{8} & \num{8} & \num{1} & \SI{0.0}{\second} \\
      10 & 5 & \num{24}  & \num{-1356} & \num{24} & \num{41} & \SI{1.5}{\second}   & \num{0} & \num{24} & \num{388} & \SI{0.4}{\second}   & \num{12} & \num{24} & \num{49} & \SI{0.5}{\second} \\
      10 & 8 & \num{28}  & \num{-5004} & \num{28} & \num{77} & \SI{2.4}{\second}   & \num{0} & \num{28} & \num{535} & \SI{0.6}{\second}   & \num{12} & \num{28} & \num{57} & \SI{0.7}{\second} \\
      10 & 10 & \num{13}  & \num{-3795} & \num{13} & \num{147} & \SI{3.5}{\second}   & \num{0} & \num{13} & \num{1016} & \SI{0.6}{\second}   & \num{5} & \num{13} & \num{1} & \SI{0.3}{\second} \\
      15 & 2 & \num{14}  & \num{14} & \num{14} & \num{0} & \SI{0.7}{\second}   & \num{0} & \num{14} & \num{155} & \SI{0.1}{\second}   & \num{14} & \num{14} & \num{1} & \SI{0.4}{\second} \\
      15 & 4 & \num{24}  & \num{-888} & \num{24} & \num{189} & \SI{1.6}{\second}   & \num{0} & \num{24} & \num{1783} & \SI{1.8}{\second}   & \num{24} & \num{24} & \num{2} & \SI{0.3}{\second} \\
      15 & 8 & \num{88}  & \num{-13344} & \num{88} & \num{1042} & \SI{11.5}{\second}   & \num{0} & \num{88} & \num{13404} & \SI{17.0}{\second}   & \num{32} & \num{88} & \num{118} & \SI{2.8}{\second} \\
      15 & 11 & \num{89}  & \num{-26575} & \num{89} & \num{1067} & \SI{19.3}{\second}   & \num{0} & \num{89} & \num{16118} & \SI{22.2}{\second}   & \num{25} & \num{89} & \num{642} & \SI{5.4}{\second} \\
      15 & 15 & \num{15}  & \num{-15421} & \num{15} & \num{2912} & \SI{96.3}{\second}   & \num{0} & \num{15} & \num{34796} & \SI{22.5}{\second}   & \num{7} & \num{15} & \num{1906} & \SI{11.0}{\second} \\
      20 & 3 & \num{18}  & \num{18} & \num{18} & \num{1} & \SI{3.4}{\second}   & \num{0} & \num{18} & \num{738} & \SI{0.6}{\second}   & \num{18} & \num{18} & \num{1} & \SI{0.1}{\second} \\
      20 & 5 & \num{64}  & \num{-3616} & \num{64} & \num{2484} & \SI{20.0}{\second}   & \num{0} & \num{64} & \num{52716} & \SI{60.2}{\second}   & \num{32} & \num{64} & \num{775} & \SI{6.9}{\second} \\
      20 & 10 & \num{199}  & \num{-41745} & \num{199} & \num{2463} & \SI{136.8}{\second}   & \num{0} & \num{199} & \num{291086} & \SI{1023.7}{\second}   & \num{55} & \num{199} & \num{1153} & \SI{24.9}{\second} \\
      20 & 15 & \num{170}  & \num{-92526} & \num{170} & \num{15029} & \SI{574.0}{\second}   & \num{0} & \num{170} & \num{383481} & \SI{1502.9}{\second}   & \num{42} & \num{170} & \num{10020} & \SI{205.3}{\second} \\
      20 & 20 & \num{26}  & \num{-39890} & \num{26} & \num{346254} & \SI{4096.1}{\second}   & \num{0} & \num{26} & \num{1100186} & \SI{847.1}{\second}   & \num{10} & \num{26} & \num{13454} & \SI{131.1}{\second} \\
      25 & 3 & \num{23}  & \num{23} & \num{23} & \num{1} & \SI{0.5}{\second}   & \num{0} & \num{23} & \num{2020} & \SI{4.0}{\second}   & \num{23} & \num{23} & \num{1} & \SI{0.3}{\second} \\
      25 & 6 & \num{140}  & \num{-10580} & \num{140} & \num{63868} & \SI{346.9}{\second}   & \num{0} & \num{140} & \num{1453552} & \SI{4113.0}{\second}   & \num{60} & \num{140} & \num{7834} & \SI{90.7}{\second} \\
      25 & 13 & \num{302}  & \num{-123422} & \num{302} & \num{414502} & \SI{4139.6}{\second}   & \num{0} & \num{0} & \num{608642} & \faClockO   & \num{78} & \num{302} & \num{4813} & \SI{176.2}{\second} \\
      25 & 19 & \num{335}  & \num{-236145} & \num{-8997} & \num{628159} & \faClockO   & \num{0} & \num{0} & \num{697217} & \faClockO   & \num{63} & \num{335} & \num{298622} & \SI{4231.8}{\second} \\
      25 & 25 & \num{36}  & \num{-81916} & \num{-9398} & \num{330251} & \faClockO   & \num{0} & \num{0} & \num{12052530} & \faClockO   & \num{12} & \num{36} & \num{140836} & \SI{818.7}{\second} \\
      30 & 4 & \num{54}  & \num{-1998} & \num{54} & \num{14945} & \SI{100.5}{\second}   & \num{0} & \num{54} & \num{762946} & \SI{1947.4}{\second}   & \num{54} & \num{54} & \num{5} & \SI{1.3}{\second} \\
      30 & 8 & \num{268}  & \num{-38364} & \num{268} & \num{1473990} & \SI{8256.5}{\second}   & \num{0} & \num{0} & \num{424215} & \faClockO   & \num{92} & \num{268} & \num{117667} & \SI{1162.5}{\second} \\
      30 & 15 & \num{496}  & \num{-246736} & \num{-12523} & \num{636725} & \faClockO   & \num{0} & \num{0} & \num{501220} & \faClockO   & \num{112} & \num{496} & \num{87424} & \SI{2941.5}{\second} \\
      30 & 23 & \num{544}  & \num{-501512} & \num{-75545} & \num{96704} & \faClockO   & \num{0} & \num{0} & \num{563463} & \faClockO   & \num{88} & \num{96} & \num{267257} & \faClockO \\
      30 & 30 & \num{59}  & \num{-146285} & \num{-57671} & \num{20221} & \faClockO   & \num{0} & \num{0} & \num{8282519} & \faClockO   & \num{15} & \num{29} & \num{1613292} & \faClockO \\
      35 & 4 & \num{64}  & \num{-2368} & \num{64} & \num{19874} & \SI{159.2}{\second}   & \num{0} & \num{33} & \num{2174070} & \faClockO   & \num{64} & \num{64} & \num{21} & \SI{3.9}{\second} \\
      35 & 9 & \num{400}  & \num{-69660} & \num{-3840} & \num{1281756} & \faClockO   & \num{0} & \num{0} & \num{405240} & \faClockO   & \num{108} & \num{400} & \num{689174} & \SI{8432.0}{\second} \\
      35 & 18 & \num{970}  & \num{-508734} & \num{-71146} & \num{108013} & \faClockO   & \num{0} & \num{0} & \num{512088} & \faClockO   & \num{162} & \num{213} & \num{173793} & \faClockO \\
      35 & 26 & \num{930}  & \num{-928590} & \num{-351158} & \num{18151} & \faClockO   & \num{0} & \num{0} & \num{576608} & \faClockO   & \num{130} & \num{130} & \num{80653} & \faClockO \\
      35 & 35 & \num{85}  & \num{-237711} & \num{-121870} & \num{16103} & \faClockO   & \num{0} & \num{0} & \num{6416944} & \faClockO   & \num{17} & \num{18} & \num{780767} & \faClockO
    \end{tabular}
  \end{center}
\end{table}

\end{landscape}

We also tried to determine a small formulation for problem~\eqref{model_bernasconi} for $\G = \C$, i.e., for the family of potentially complemented products of variables.
To this end, system~\eqref{eq_system} can be augmented by binary variables to indicate whether a $g \in \G$ is used (that is, has a nonzero multiplier).
The minimization of the sum of these binary variables yields $\lc{\G}(f)$.
Clearly, this approach only works for small sizes $|\G|$.
Our results for $N \in \{3,4,5,6\}$ and $R = N$ were quite disappointing -- at least when considering only functions $g \in \C$ of degree at most $5$, the linearization complexity is minimized by the linearization that just uses monomials.
Hence, we do not expect that one can gain much by allowing complements of variables for the low auto-correlation problem.

\section{Open Problems}
\label{sec_open_problems}

We hope to have convinced the reader that the linearization complexity is a useful concept.
Nevertheless, many unsolved problems remain, and we use the rest of the paper to present them.

\paragraph{More families.}
We discussed various families of functions $g$ to use for linearization, namely the products of variables $\M$, the potentially complemented products of variables $\C$, arbitrary Boolean functions $\B$, and finally functions that indicate whether an expression of a certain type attains a certain value.
Another family is induced by any quadratization strategy: after applying such a strategy to obtain a quadratization function $g$ as in~\eqref{eq_quadratization} one can linearize the latter, e.g., in a monomial-wise fashion.
We believe that there are more such interesting families.

\paragraph{Formulations.}
Perfect formulations are known for the first two considered linearizations when considering every linearization variables separately.
Strengthening the joint formulation for multiple linearization variables from $\C$ is subject of current research, e.g., by means of studying multilinear polytopes~\cite{BuchheimCR19,CramaR17,DelPiaD21,DelPiaK17,DelPiaK18,DelPiaK18a,DelPiaK21,DelPiaW22,Rodriguez-Heck18}.
For $\B$ we cannot hope to identify perfect formulations since this encompasses arbitrary binary sets.
However, for other functions it is interesting to investigate which inequalities are best to add in order to apply such a function.
For instance, our formulation~\eqref{ip_indicators} worked well because we did not model the meaning of each $z$-variable individually, but because we considered multiple of them in a combined fashion in constraints~\eqref{ip_indicators_convexity} and~\eqref{ip_indicators_indicate}.

\paragraph{Bounding techniques.}
While we could interpret the linearization complexities for $\M$ and for $\B$, little is known about $\C$.
The first natural question is which properties of a polynomial allow to recast it as a sum of only few products of potentially complemented variables.

\paragraph{Algorithmic questions.}
After settling how $f$ could be encoded (expressing it as a polynomial is only one possibility), there are many algorithmic problems related to linearization complexity.
Most importantly, the complexity of the computation or approximation of $\lc{\C}(f)$ or of $\lc{\B}(f)$ is open.
For practical purposes it would be very interesting to find small linearizations based on $\C$ because the actual formulations are essentially the same as those for $\M$ which are reasonably well understood.

\paragraph{Approximations.}
While \cref{thm_null_set} and \cref{thm_random_polynomial} indicate that linearizations with small linearization complexity are rare, one may still consider approximate linearizations, i.e., small linearizations of a function $f'$ that is very close to $f$.
Instead of abandoning exactness one can also try to pursue a related approach by finding a linearization that may not be small but for which only few of the weights $b_i$ in~\eqref{eq_linearization} are large (in absolute value).
While the resulting IP formulation would still have many variables, one could apply strengthening techniques only on those few linearizations that are most important for the value of $f$.
Then, relaxation errors for the remaining variables (with small $|b_i|$) will not have a big impact on the overall objective value, which would yield better bounds.

\paragraph{Acknowledgements.}
We thank three anonymous referees for their constructive comments that led to an improved presentation of the paper. In fact, some of the results would not have been obtained without their valuable feedback.

\bibliographystyle{plainurl}
\bibliography{binary-linearization-complexity}

\end{document}